\documentclass[12pt]{article}
\usepackage[utf8]{inputenc}
\usepackage[a4paper,width=170mm,top=15mm,bottom=20mm]{geometry}
\usepackage{newtxmath}

\usepackage{amsfonts,amssymb}

\DeclareMathOperator*{\argminA}{arg\,min}

\usepackage{graphicx}
\usepackage{subcaption}
\usepackage{amsthm}
\newtheorem{corollary}{Corollary}
\usepackage{hyperref}
\usepackage{xcolor}
\hypersetup{
    colorlinks,
    linkcolor={red!50!black},
    citecolor={blue!50!black},
    urlcolor={blue!80!black}
}
\usepackage[round]{natbib}
\usepackage{float}
\usepackage{authblk}
\usepackage{orcidlink}
\usepackage{setspace}

\title{
On the Selection Stability of Stability Selection and Its Applications}

\author[1]{Mahdi Nouraie\orcidlink{0000-0002-4792-4994}}
\author[1,2]{Samuel Muller\orcidlink{0000-0002-3087-8127}\thanks{Address for correspondence: samuel.muller@mq.edu.au}}
\affil[1]{School of Mathematical and Physical Sciences, Macquarie University}
\affil[2]{School of Mathematics and Statistics, The University of Sydney}

\date{}
\begin{document}
\begin{spacing}{1}
\maketitle
\begin{abstract}
\noindent{Stability selection is a widely adopted resampling-based framework for high-dimensional variable selection. This paper seeks to broaden the use of an established stability estimator to evaluate the overall stability of the stability selection results, moving beyond single-variable analysis. We suggest that the stability estimator offers two advantages: it can serve as a reference to reflect the robustness of the results obtained, and it can help identify a Pareto optimal regularization value to improve stability. By determining the regularization value, we calibrate key stability selection parameters, namely, the decision-making threshold and the expected number of falsely selected variables, within established theoretical bounds. In addition, the convergence of stability values over successive sub-samples sheds light on the required number of sub-samples addressing a notable gap in prior studies. The \texttt{stabplot} R package is developed to facilitate the use of the methodology featured in this paper.}
\end{abstract}
\end{spacing}

\section{Introduction}\label{s1}
Stability selection, a widely recognized resampling-based variable selection framework, performs selection by examining how frequently variables are chosen across multiple randomly selected resamples \citep{meinshausen2010stability}. \citet{meinshausen2010stability} analyzed the selection frequencies for individual variables over a range of regularization values, presenting the results in a plot known as the `stability path'. This plot facilitates the identification of variables that are consistently selected over the regularization grid. In this paper, we shift the focus from examining the stability of individual variables to evaluating the overall stability of the stability selection results. Employing an established stability estimator, we assess the stability of the results in its entirety and introduce what we call `stable stability selection', which identifies the variables that are associated with highly stable results.

We consider a dataset $\mathcal{D} = \{(\boldsymbol{x}^\top_{i}, y_{i})\}_{i=1}^{n}$, where each element consists of a univariate response $y_{i} \in \mathbb{R}$ and a $p$-dimensional vector of fixed covariates $\boldsymbol{x}^\top_{i} \in \mathbb{R}^{p}$. In the context of linear regression, it is typically assumed that pairs $(\boldsymbol{x}^\top_{i}, y_{i})$ are independent and identically distributed (i.i.d.); when these are assumed to be random, but for simplicity, we assume that the covariate vector is fixed.

The linear regression model is formally expressed as $\boldsymbol{Y} = \beta_0 + X\boldsymbol{\beta} + \boldsymbol{\varepsilon}$, where $X \in \operatorname{mat}(n \times p)$ represents the design matrix, $\boldsymbol{Y} \in \mathbb{R}^{n}$ is the $n$-dimensional vector corresponding to the univariate response variable, and $\beta_0 \in \mathbb{R}$ denotes the intercept term. The vector of regression coefficients for the $p$ non-constant covariates is denoted by $\boldsymbol{\beta} \in \mathbb{R}^{p}$, while $\boldsymbol{\varepsilon} \in \mathbb{R}^{n}$ represents the $n$-dimensional vector of random errors. It is assumed that the components of $\boldsymbol{\varepsilon}$ are i.i.d. and independent of the covariates.

Variable selection, in accordance with the terminology of \citet{meinshausen2010stability}, typically involves categorizing covariates into two distinct groups: the signal group $S := \{k \neq 0 \mid \beta_k \neq 0\}$ and the noise group $N := \{k \neq 0 \mid \beta_k = 0\}$, where $S \cap N = \emptyset$. The primary objective of variable selection is to accurately identify the signal group $S$.

The stability of a variable selection method refers to the consistency with which it selects variables across different training sets drawn from the same underlying distribution \citep{kalousis2007stability}. The stability selection approach described by \citet{meinshausen2010stability} uses sub-samples that are half the size of the sample size of the original dataset to perform variable selection. Stability selection enables the identification of variables that are consistently selected as relevant across the majority of sub-samples given a regularization value, considering them as stable variables. Furthermore, \citet{meinshausen2010stability} established an asymptotic upper-bound for the Per-Family Error Rate, which represents the expected number of falsely selected variables. Later, \citet{shah2013variable} introduced the complementary pairs stability selection method, which computes selection frequencies by counting how often a variable is included in models fitted to complementary $50\%$ sub-samples. \citet{shah2013variable} introduced stricter upper-bounds for the expected number of variables selected with low selection frequencies.

There is a substantial body of research on the stability of variable selection methods, where numerous stability measures have been proposed. For more comprehensive information, we refer to \citet{kuncheva2007stability}, \citet{nogueira2018stability}, and \citet{sen2021critical}. \citet{nogueira2018stability} conducted an extensive literature review on the topic, consolidating the desirable properties of stability measures into five mathematical conditions. \citet{nogueira2018stability} showed that none of the stability measures previously introduced satisfied all five conditions. In response, \citet{nogueira2018stability} proposed a novel stability measure, which was the first in the literature to meet all these criteria, generalizing several previous works such as in \citet{kuncheva2007stability}.

\citet{nogueira2018stability}, linking their work with \citet{gwet2008variance}, demonstrated that as the number of resamples approaches infinity, their stability measure, $\hat{\Phi}(\cdot)$, converges to a Normal distribution with a location parameter $\mu = \Phi$, where $\Phi$ denotes population stability.  They also compared the stability of the stability selection with that of the Least Absolute Shrinkage and Selection Operator \citep[Lasso;][]{tibshirani1996regression}, concluding that the former exhibits greater stability than the latter in most of the scenarios considered.

In this paper, we apply the stability estimator proposed by \citet{nogueira2018stability} to the stability selection results, and proxify the stability values to find the optimal regularization value that achieves high stability with the least possible loss in prediction accuracy. The optimal regularization value then is employed to calibrate two critical parameters of stability selection: the decision-making threshold and the expected number of falsely selected variables, balancing them with respect to one another. We also use the optimal regularization value to define a new variable selection criterion, called stable stability selection, prioritizing the selection of variables associated with highly stable results, rather than those mostly selected in their best-case scenarios over the regularization grid, which is the primary criterion in the traditional stability selection framework. In addition, the convergence value of the stability estimator over the successive sub-samples serves as a reference to assess the reliability of the results obtained. The point at which convergence occurs provides valuable insight into the required number of sub-samples, an aspect for which we found no dedicated research addressing its determination.

\citet{nogueira2018stability} suggested that stability and prediction accuracy can be interpreted through the lens of a Pareto front \cite{pareto1896cours}. We refer to this methodology as stability-accuracy selection. Building on this perspective, we demonstrate that, given two justifiable assumptions, our regularization value constitutes a Pareto-optimal solution within this context.

Hyper-parameter tuning can be approached from various perspectives \citep{Aftab03052025}.
The use of stability to enhance variable selection has been a subject of interest since at least the work of \citet{NIPS2010_301ad0e3}. Among the existing literature, the study by \citet{10.5555/2567709.2567772} is perhaps the most closely related to the present work and has continued to attract interest, as demonstrated by contributions such as \citet{Yang03042020}, \citet{reavie2021pareto}, \citet{Wen03072023}, and \citet{lee2024determining}. In \citet{10.5555/2567709.2567772}, the regularization parameter of penalized regression models is selected based on a stability-driven criterion. Nonetheless, our methodology departs from that of \citet{10.5555/2567709.2567772} by being fully embedded within the stability selection framework, using its internal selection outcomes to guide regularization tuning. In contrast, their method is designed to function externally, resulting in greater computational overhead due to its decoupled implementation.

The rest of this paper is organized as follows. Section \ref{s2} outlines the proposed methodology. Section \ref{s3} describes the real and synthetic datasets used in the paper. The numerical results of applying the proposed method, along with the demonstration of Pareto optimality of the proposed regularization value, are presented in Section \ref{s4}.

%
%
%
%
\section{Methodology}\label{s2}
In this section, we present a methodology to evaluate the selection stability of the stability selection results and to determine the optimal regularization value accordingly. Although we demonstrate the proposed method within the context of stability selection outlined by \citet{meinshausen2010stability}, it is important to note that the methodology is adaptable to other formulations of stability selection, such as those described by \citet{shah2013variable}. We will introduce our methodology following a brief overview of the foundational approach established by \citet{meinshausen2010stability}.

\subsubsection*{Stability Selection}

The stability selection framework requires the implementation of an appropriate selection method. One commonly used technique for variable selection in the linear regression context is the Lasso estimator,  formally defined as
\begin{equation*}\label{eqn:Lasso}
\hat{\beta}_0(\lambda),\boldsymbol{\hat{\beta}}(\lambda) = \argminA_{\beta_0 \in \mathbb{R}, \boldsymbol{\beta} \in \mathbb{R}^{p}} \left(\|\boldsymbol{Y} - \beta_0 - X\boldsymbol{\beta}\|_{2}^{2} + \lambda \sum_{k = 1}^{p} |\beta_k|\right),
\end{equation*}
where $\lambda \in \mathbb{R}^{+}$ represents the Lasso regularization parameter. The set of non-zero coefficients can be identified as $\hat{S}(\lambda) := \{k \neq 0 \mid \hat{\beta}_{k}(\lambda) \neq 0\}$, derived from solving the Lasso equation through convex optimization.

\citet{meinshausen2010stability} defined the stable set as 
\begin{equation}\label{eqn:S_stable}
    \hat{S}^{\text{stable}} := \{j \mid \max_{\lambda \in \Lambda} (\hat{\Pi}_{j}^{\lambda}) \geq \pi_{\text{thr}}\};\quad j = 1, \ldots p,
\end{equation}
where $\Lambda$ represents the set of regularization values, $\pi_{\text{thr}} \in [0.6,0.9]$ is the threshold for decision-making in variable selection, and $\hat{\Pi}_{j}^{\lambda}$ denotes the selection frequency of the $j$th variable given the regularization parameter $\lambda$. Equation~\eqref{eqn:S_stable} identifies variables whose selection frequencies exceed $\pi_{\text{thr}}$ under the regularization value that maximizes their selection frequencies; therefore, it considers best-case of each variable for decision-making. This approach has been criticized in subsequent studies, notably by \citet{zhou2013patient}.

\subsubsection*{Stable Stability Selection}

We first define a grid of regularization values, $\Lambda$, from which the optimal $\lambda$ for Lasso is to be identified. For this purpose, we use as the default choice the $\lambda$ values generated by the \texttt{cv.glmnet} function from the \texttt{glmnet} package in R \citep{friedman2010regularization}, applying a 10-fold cross-validation to the complete dataset $\mathcal{D}$. Alternative regularization values can also be employed. However, a comparison of different methods for generating the regularization grid falls outside the scope of this paper. The goal is to identify the optimal regularization value $\lambda_{\text{stable}} \in \Lambda$ that yields highly stable outcomes with the least possible loss in terms of predictive ability.

The following procedure is applied to each regularization value $\lambda \in \Lambda$. In each iteration of the stability selection process, a sub-sample, comprising half the size of the original dataset $\mathcal{D}$, is randomly selected. The Lasso model is then fitted to this sub-sample using the current value of $\lambda$, and the binary selection outcomes are recorded as a row in the binary selection matrix $M(\lambda) \in \operatorname{mat}(B \times p)$, where $B$ denotes the number of sub-samples. Therefore, in the end, we obtain $|\Lambda|$ distinct selection matrices $M(\lambda)$. In this context, $M(\lambda)_{bj} = 1$ indicates that the $j$th variable is identified as part of the signal set $\hat{S}(\lambda)$ when Lasso is applied to the $b$th sub-sample given $\lambda$. In contrast, $M(\lambda)_{bj} = 0$ signifies that the $j$th variable is classified as belonging to the noise set $\hat{N}(\lambda)$ after applying Lasso to the $b$th sub-sample given $\lambda$.

The stability estimator proposed by \citet{nogueira2018stability} is defined as
\begin{equation}\label{eqn: phi}
    \hat{\Phi}(M(\lambda)) := 1 - \frac{\frac{1}{p}\sum_{j=1}^{p} s_{j}^{2}}{\frac{q(\lambda)}{p}\left(1 - \frac{q(\lambda)}{p}\right)};\quad \lambda \in \Lambda,
\end{equation}
where $s_{j}^{2}$ denotes the unbiased sample variance of the binary selection statuses of the $j$th variable, while $q(\lambda)$ denotes the average number of variables selected under the regularization parameter $\lambda$. The stability estimator $\hat{\Phi}(\cdot)$ is bounded by $\left[-\frac{1}{B-1}, 1\right]$ \citep{nogueira2018stability}; the larger the $\hat{\Phi}(M(\lambda))$, the more stable is $M(\lambda)$. To interpret the stability values of Equation~\eqref{eqn: phi}, \citet{nogueira2018stability} adopted the guidelines proposed by \citet{fliess2004measurment}. According to \citet{nogueira2018stability}, stability values that exceed $0.75$ indicate excellent agreement between selections beyond what would be expected by random chance, while values below $0.4$ signify poor agreement between them. Based on \citet{nogueira2018stability}, stability values that reside between these thresholds are categorized as indicative of intermediate to good stability.

By obtaining all $|\Lambda|$ selection matrices, we can estimate the stability of each using the formulation given in Equation~\eqref{eqn: phi}. We propose that the optimal regularization value, denoted as $\lambda_{\text{stable}}$, is the smallest regularization value at which the stability measure surpasses 0.75, that is, 
\begin{equation}\label{eqn: lambda_stable}
    \lambda_{\text{stable}} := \min\left\{\lambda \in \Lambda \mid \hat{\Phi}(M(\lambda)) \geq 0.75\right\}.
\end{equation}
The threshold value of $0.75$ is somewhat arbitrary and may not be attainable in certain situations. As we will see in Section \ref{s4}, $\lambda_{\text{stable}}$ may not exist in certain practical applications due to instability of the results, which prevents the stability values from exceeding $0.75$. In such cases, aligning with the $\lambda_{\text{1se}}$ strategy introduced by \citet{hastie2009elementss}, we propose 
\begin{equation}\label{eqn: lambda_stable_1se}
    \lambda_{\text{stable-1sd}} := \min\left\{\lambda \in \Lambda \mid \hat{\Phi}(M(\lambda)) \geq \max_{\lambda \in \Lambda} \hat{\Phi}(M(\lambda)) - \text{SD}_{\lambda \in \Lambda}\left(\hat{\Phi}(M(\lambda))\right)\right\},
\end{equation}
where SD$(\cdot)$ represents the standard deviation function.

Having $\lambda_{\text{stable}}$ or $\lambda_{\text{stable-1sd}}$ allows a fresh perspective on variable selection through stability selection. \citet{meinshausen2010stability} introduces Equation~\eqref{eqn:S_stable} to identify stable variables, that is, those with selection frequencies that exceed $\pi_{\text{thr}}$ in the best-case scenario. Instead, we propose
\begin{equation}\label{eqn:S_stable_2}
    \hat{S}^{\text{stable}} := \left\{j \mid \hat{\Pi}_{j}^{\lambda_{\text{stable}}} \geq \pi_{\text{thr}}\right\};\quad j = 1, \ldots p.
\end{equation}
Equation~\eqref{eqn:S_stable_2} represents the variables with selection frequencies higher than $\pi_{\text{thr}}$ under the $\lambda_{\text{stable}}$, which corresponds to the highly stable outcomes. We term this selection criterion `stable stability selection' because it leverages highly stable results to facilitate variable selection. If $\lambda_{\text{stable}}$ does not exist, $\hat{S}^{\text{stable-1sd}}$ can be defined similarly.

Both upper-bounds of the Per-Family Error Rate, mentioned in Section \ref{s1}, provided by \citet{meinshausen2010stability} and \citet{shah2013variable}, depend on the values of $\lambda$ and $\pi_{\text{thr}}$. The calibration of these two values requires the arbitrary selection of one of these two parameters, which can be challenging to justify \citep{bodinier2023automated}. By adopting $\lambda = \lambda_{\text{stable}}$ or $\lambda = \lambda_{\text{stable-1sd}}$, the upper-bound can be explicitly tailored to the data owner's preferences. If a pre-determined $\pi_{\text{thr}}$ is chosen, the Per-Family Error Rate is deterministically obtained; alternatively, if a fixed pre-determined Per-Family Error Rate value is desired, the corresponding $\pi_{\text{thr}}$ is enforced.

The asymptotic upper-bound provided by \citet{meinshausen2010stability} for the Per-Family Error Rate is given by
\begin{equation}\label{eqn:MBPFER}
    \text{PFER}(\Lambda, \pi_{\text{thr}}) = \frac{1}{2\pi_{\text{thr}} - 1}\frac{q_{\Lambda}^2}{p},
\end{equation}
where $q_{\Lambda}$ represents the average number of selected variables over the regularization grid $\Lambda$. In a manner consistent with the approach of \citet{bodinier2023automated}, we can derive a point-wise control version of Equation~\eqref{eqn:MBPFER} by restricting $\Lambda$ to a single value, $\lambda$, without affecting the validity of the original equation. Consequently, we can replace $q_{\Lambda}$ with $q(\lambda)$, as introduced in Equation~\eqref{eqn: phi}. By employing $\lambda_{\text{stable}}$ within this formula and utilizing $q(\lambda_{\text{stable}})$, we establish a two-way control: fixing $\pi_{\text{thr}}$ allows us to determine the upper-bound, and vice versa. A similar approach can be applied to the upper-bounds proposed by \citet{shah2013variable}.

As outlined in Section \ref{s1}, as the number of sub-samples increases, the convergence of $\hat{\Phi}(\cdot)$ to the population stability $\Phi$ is ensured. To illustrate this, we define $M^{(t)}(\lambda)$ as the matrix containing the selection results from the first $t$ sub-samples for a given $\lambda \in \Lambda$, that is, the first $t$ rows of $M(\lambda)$. The objective is to evaluate the stability of $M^{(t)}(\lambda)$ across successive sub-samples in order to determine the appropriate cut-off point for the process, that is, the number of sub-samples required to achieve convergence in stability values. Beyond this threshold, additional sub-samples do not significantly change the stability of $M(\lambda)$. We propose plotting the stability values of the selection matrix over the sequential sub-sampling, that is $\hat{\Phi}(M^{(t)}(\lambda));\; t = 2, 3, \ldots B$, against $t$ to monitor the convergence status of the stability values. Given that the asymptotic distribution of $\hat{\Phi}(\cdot)$ is established, a confidence bound can be drawn along the curve to reflect the inherent uncertainty of the stability estimator.

In Section \ref{s4}, we demonstrate the Pareto optimality of $\lambda_{\text{stable}}$, and present the numerical results obtained from the application of the proposed approach to synthetic and real datasets, which are introduced in Section \ref{s3}.

\section{Datasets}\label{s3}
We evaluate our methodology by performing evaluations on synthetic and two real bioinformatics datasets, as detailed below.

\subsubsection*{Synthetic Data}

As synthetic data, we consider a dataset with a sample size of $n = 50$ that includes $p = 500$ predictor variables. The predictor variables $\boldsymbol{x}^\intercal_{i}$ are independently drawn from the Normal distribution $\mathcal{N}(0, \Sigma)$, where $\Sigma \in \operatorname{mat}(p,p)$ has a diagonal of ones and $\sigma_{jk} = \rho^{|j - k|}$ for $j \neq k$ where $\rho \in \{0.2, 0.5, 0.8\}$. The response variable is linearly dependent solely on the first two predictor variables, characterized by the coefficient vector $\boldsymbol{\beta} = (1.5, 1.1, 0, \ldots0)^{\top}$, and the error term $\boldsymbol{\varepsilon}$ is an i.i.d. sample from the standard Normal distribution $\mathcal{N}(0, 1)$.

\subsubsection*{Riboflavin Data}

For the first real example, we use the well-established `Riboflavin' dataset, which focuses on the production of riboflavin (vitamin B2) from various Bacillus subtilis. This dataset, provided by the Dutch State Mines Nutritional Products, is accessible via \texttt{hdi} R package \citep{hdi-package}. It comprises a single continuous response variable that represents the logarithm of the riboflavin production rate, alongside $p = 4,088$ covariates, which correspond to the logarithm of expression levels for 4,088 bacterial genes. The primary objective of analyzing this dataset is to identify genes that are associated with riboflavin production, with the ultimate goal of genetically engineering bacteria to enhance riboflavin yield. Data were collected from $n = 71$ relatively homogeneous samples, which were repeatedly hybridized during a feed-batch fermentation process involving different engineered strains and varying fermentation conditions. \citet{buhlmann2014high} employed stability selection with Lasso and identified three genes—\texttt{LYSC\_at}, \texttt{YOAB\_at}, and \texttt{YXLD\_at}—as relevant genes in this problem.

\subsubsection*{Affymetrix Rat Genome 230 2.0 Array}

As an additional real-world example, we investigate `Affymetrix Rat Genome 230 2.0 Array' microarray data introduced by \citet{scheetz2006regulation}. This dataset comprises $n = 120$ twelve-week-old male rats, with expression levels recorded for nearly 32,000 gene probes for each rat. The primary objective of this analysis is to identify the probes most strongly associated with the expression level of the TRIM32 probe (\texttt{1389163\_at}), which has been linked to the development of Bardet-Biedl syndrome \citep{chiang2006homozygosity}. This genetically heterogeneous disorder affects multiple organ systems, including the retina. In accordance with the pre-processing steps outlined by \citet{huang2008adaptive}, we excluded gene probes with a maximum expression level below the $25$th percentile and those exhibiting an expression range smaller than $2$. This filtering process yielded a refined set of $p = 3,083$ gene probes that demonstrated sufficient expression and variability for further analysis.

\section{Results}\label{s4}
In this section, we illustrate the efficacy of our methodology through its application to the synthetic and real datasets introduced in Section \ref{s3}. We also demonstrate in Corollary \ref{corollary1} that $\lambda_{\text{stable}}$ constitutes a Pareto optimal solution in the trade-off between stability and accuracy in variable selection. We used the GitHub repository associated with \citet{nogueira2018stability} to implement their stability estimator and to obtain confidence intervals\footnote{\url{https://github.com/nogueirs/JMLR2018}}. 

\subsubsection*{Synthetic Data}
Figure \ref{fig:three_stabilities} illustrates the selection stability of the stability selection results on a grid of regularization values, applied to the synthetic datasets introduced in Section \ref{s3}, using three different values of $\rho$ and choosing the number of sub-samples $B = 500$.  The horizontal dashed red lines indicate the stability thresholds of $0.4$ and $0.75$ as discussed in Section \ref{s2}. As illustrated in Figure \ref{fig:three_stabilities}, the regularization value that minimizes the cross-validation error, $\lambda_{\text{min}}$, falls within the poorly stable region. In addition, $\lambda_{\text{1se}}$ lies within the poor to intermediate stability regions. In contrast, $\lambda_{\text{stable}}$, introduced in Equation~\eqref{eqn: lambda_stable}, is specifically designed to fall within the region of excellent stability, where applicable.

\begin{figure}[htbp]
    \centering
    \begin{subfigure}{0.6\textwidth}
        \centering
        \includegraphics[width=\textwidth]{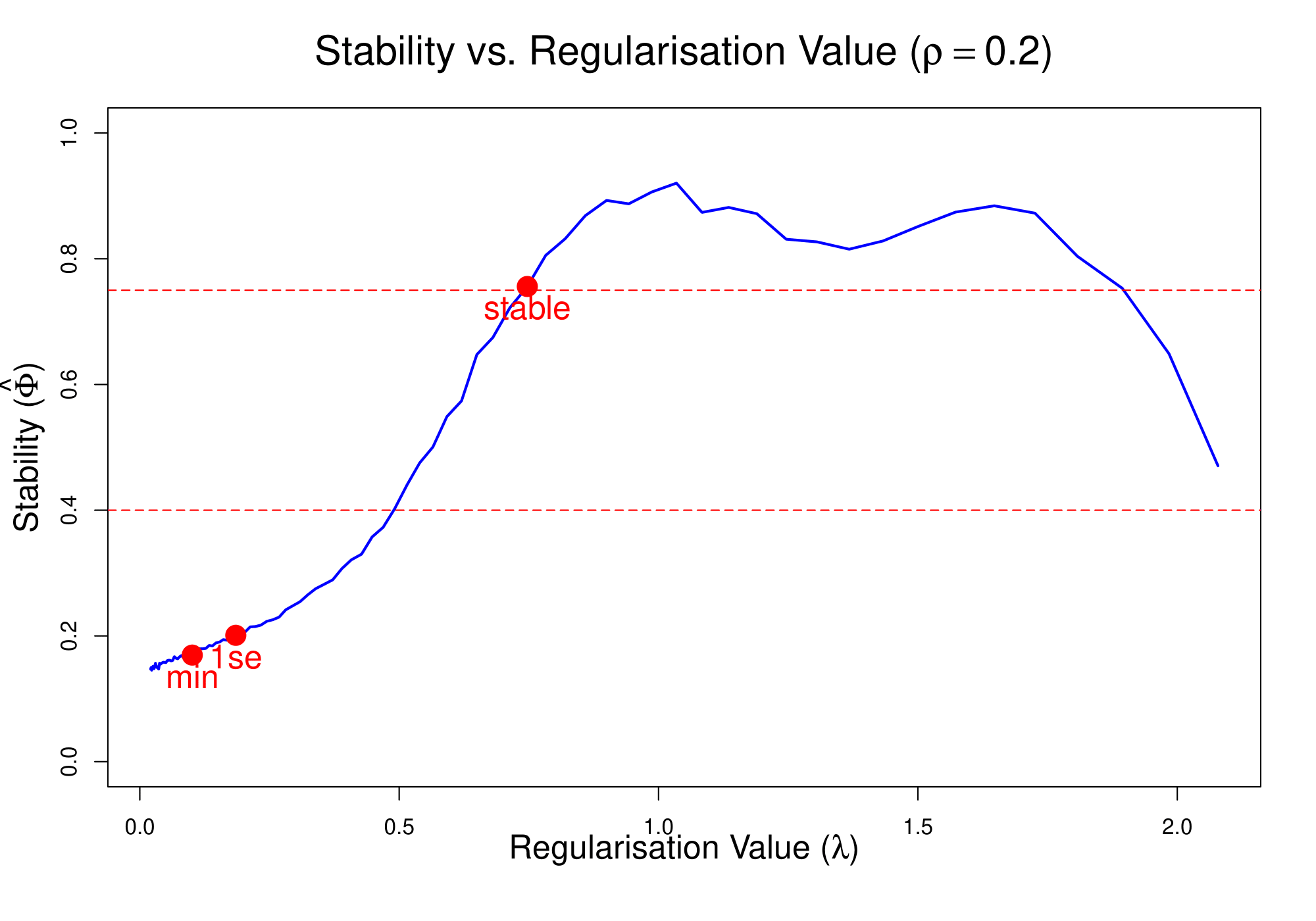}
        \caption{$\rho = 0.2$}
        \label{fig:stab1}
    \end{subfigure}
    \vfill
    \begin{subfigure}{0.6\textwidth}
        \centering
        \includegraphics[width=\textwidth]{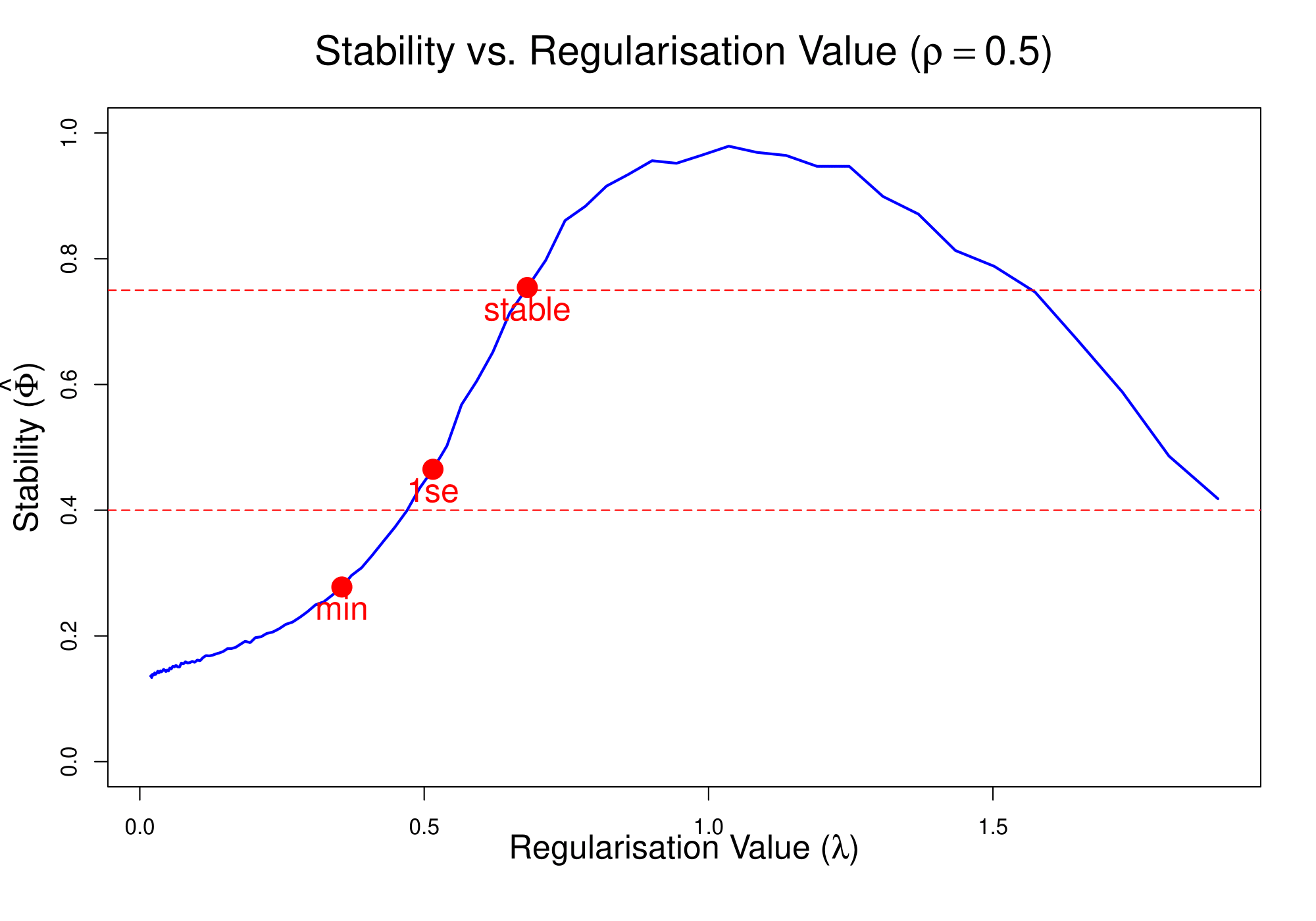}
        \caption{$\rho = 0.5$}
        \label{fig:stab2}
    \end{subfigure}
    \vfill
    \begin{subfigure}{0.6\textwidth}
        \centering
        \includegraphics[width=\textwidth]{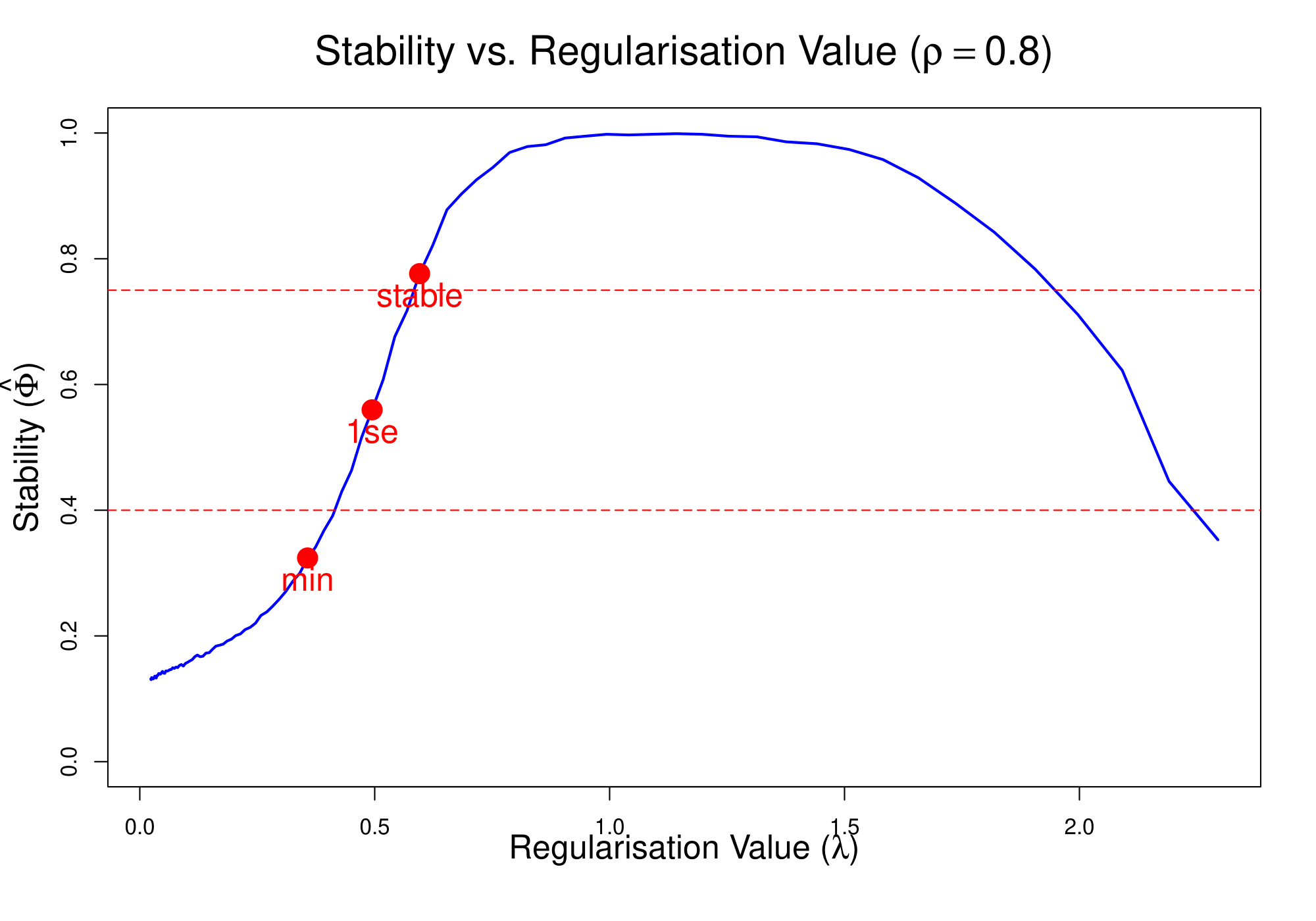} 
        \caption{$\rho = 0.8$}
        \label{fig:stab3}
    \end{subfigure}
    \caption{Selection stability of stability selection over the regularization grid on the synthetic data}
    \label{fig:three_stabilities}
\end{figure}

Notably, in all three correlation scenarios, the three regularization selection methods successfully identified relevant variables with high selection frequencies, with a minimum selection frequency of $0.994$. However, since $\lambda_{\text{stable}}$ is larger than the two other regularization values, it applies a stronger shrinkage to variables. To further examine this in terms of model accuracy, for each correlation scenario, we generate $n' = 25$ additional test samples from the corresponding distribution described in Section \ref{s3}. At each iteration of stability selection, we predict the response variable for the test samples using the estimated model and, ultimately, we aggregate all the mean squared error (MSE) values obtained over the corresponding $\lambda$ by averaging.

\begin{figure}[htbp]
    \centering
    \begin{subfigure}{0.6\textwidth}
        \centering
        \includegraphics[width=\textwidth]{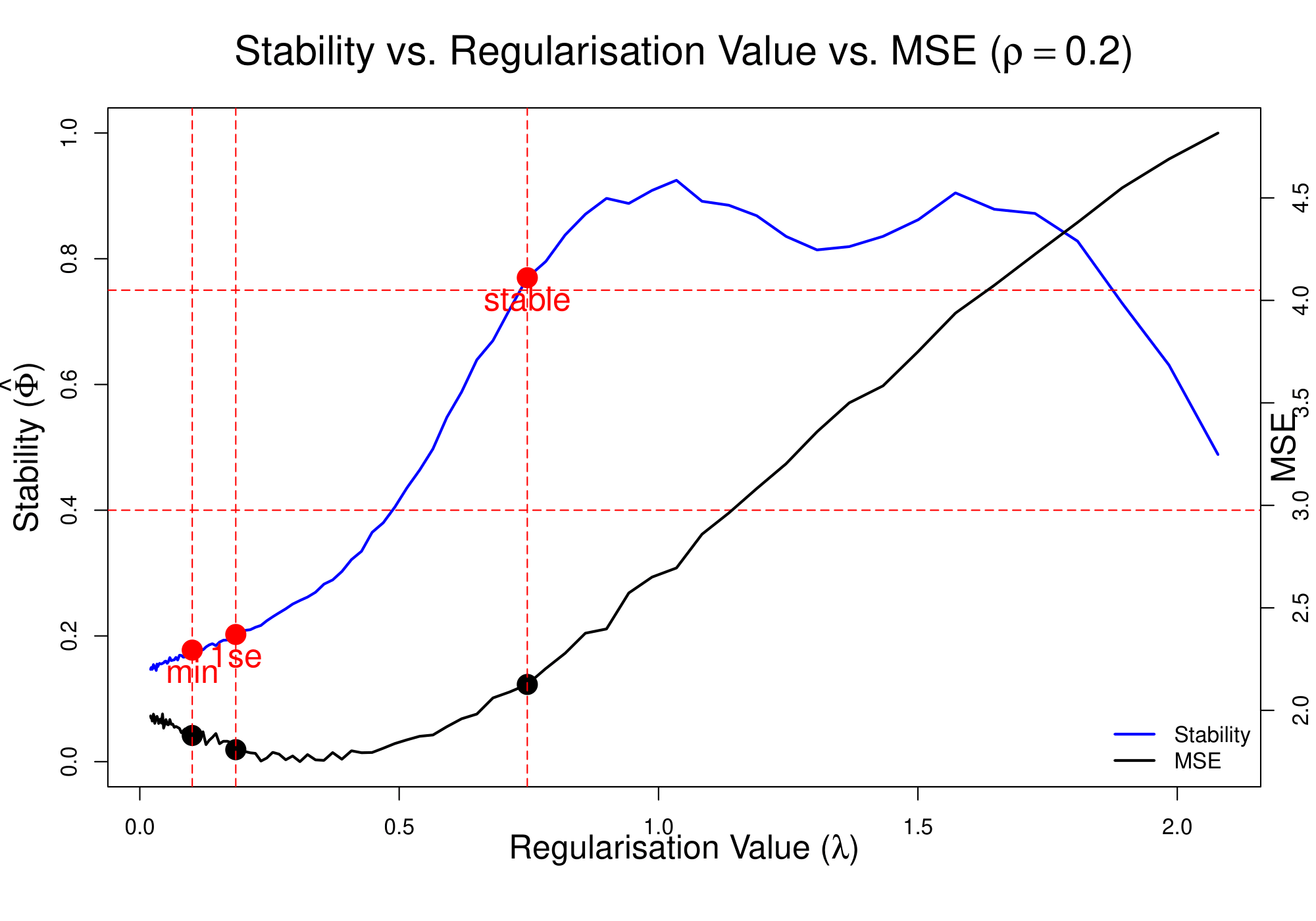}
        \caption{$\rho = 0.2$}
        \label{fig:stabMSE1}
    \end{subfigure}
    \vfill
    \begin{subfigure}{0.6\textwidth}
        \centering
        \includegraphics[width=\textwidth]{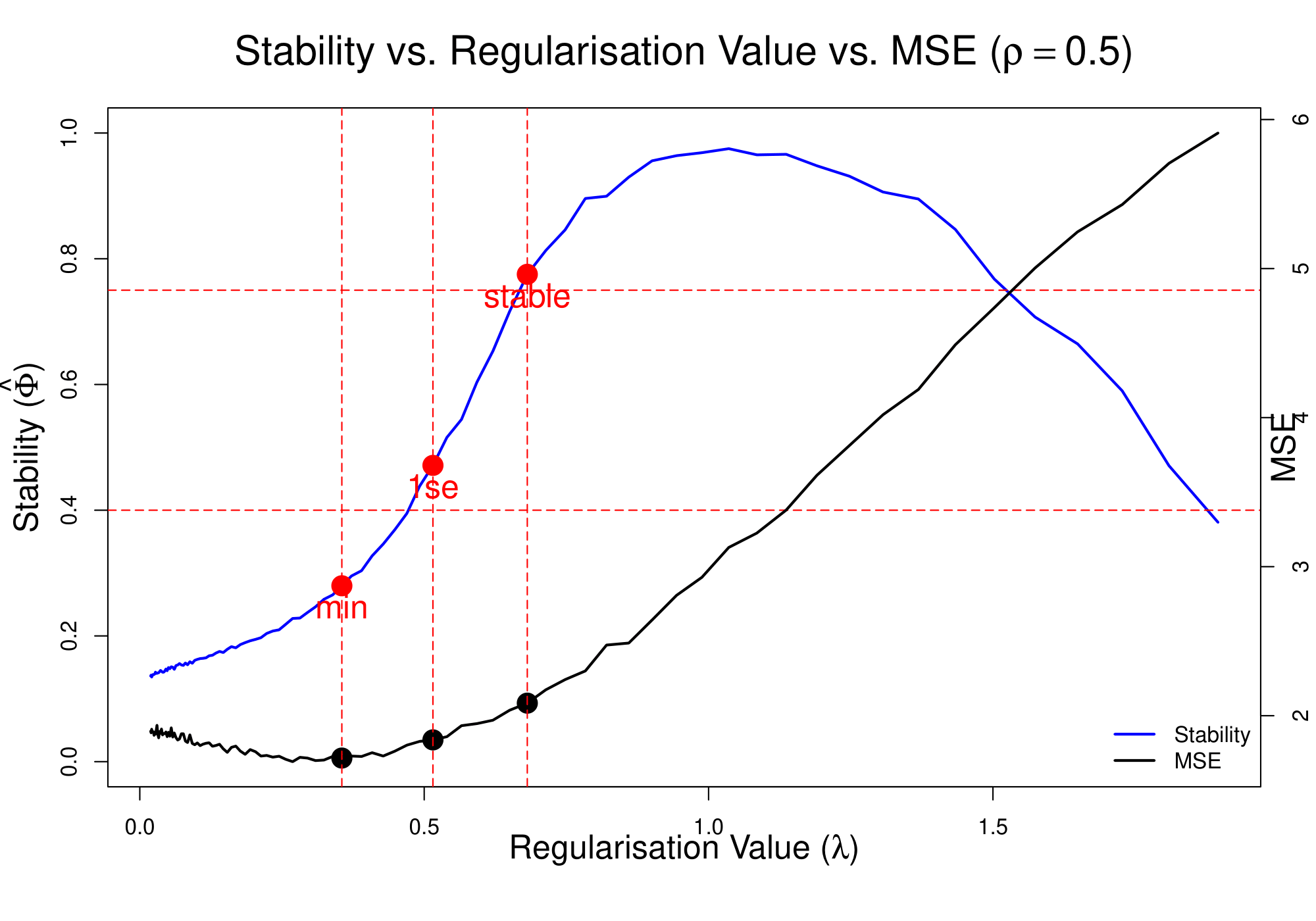}
        \caption{$\rho = 0.5$}
        \label{fig:stabMSE2}
    \end{subfigure}
    \vfill
    \begin{subfigure}{0.6\textwidth}
        \centering
        \includegraphics[width=\textwidth]{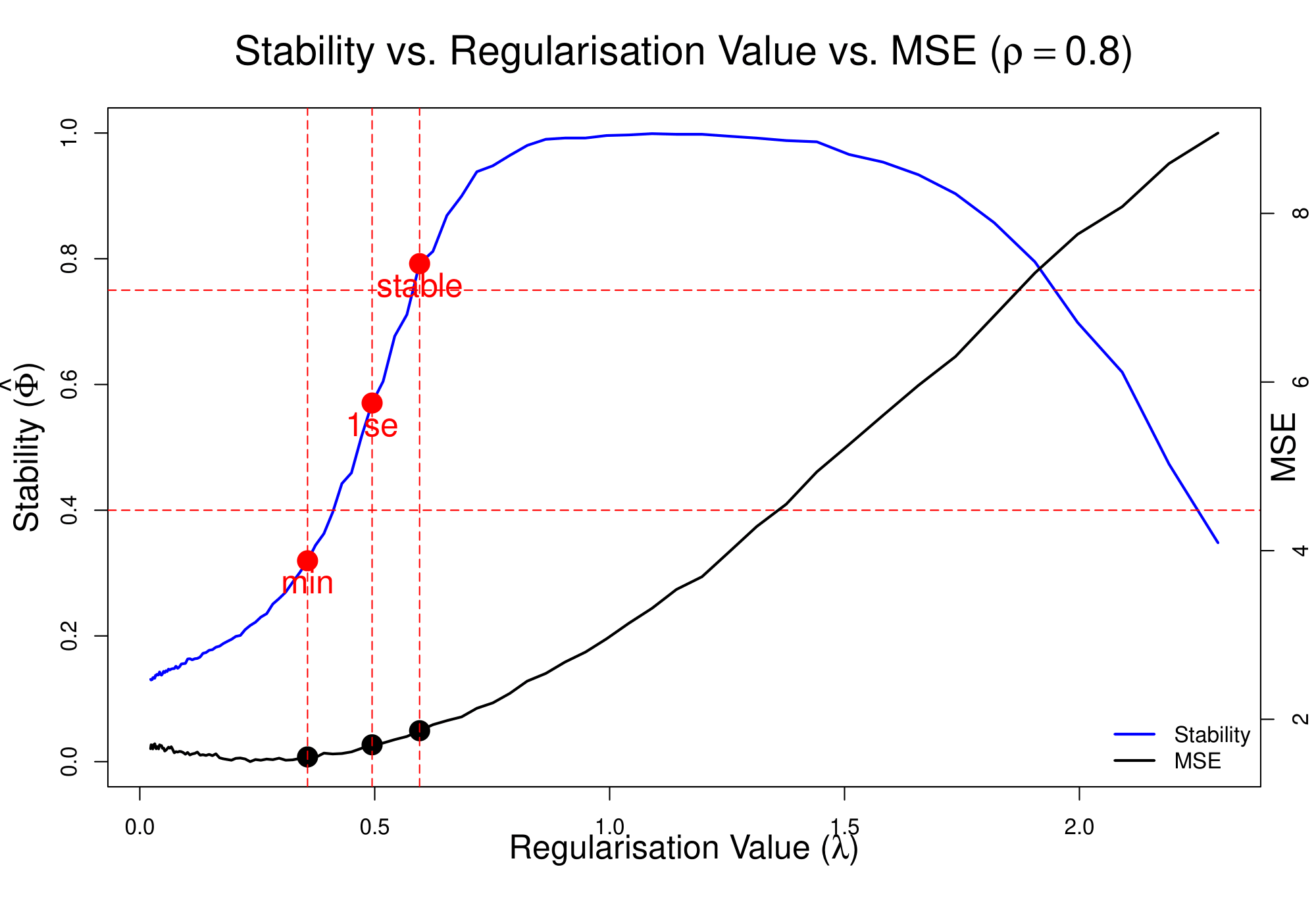} 
        \caption{$\rho = 0.8$}
        \label{fig:stabMSE3}
    \end{subfigure}
    \caption{Selection stability and MSE of stability selection over the regularization grid on the synthetic data}
    \label{fig:three_stabMSE}
\end{figure}

As demonstrated in Figure \ref{fig:three_stabMSE}, across all three correlation scenarios, the improvement in stability when transitioning from $\lambda_{\text{1se}}$ to $\lambda_{\text{stable}}$ far outweighs the reduction in accuracy. This outcome aligns closely with the findings of \citet{nogueira2018stability}, who noted that ``All these observations show that stability can potentially be increased without loss of predictive power''. 

\begin{figure}[htbp]
    \centering
    \begin{subfigure}{0.6\textwidth}
        \centering
        \includegraphics[width=\textwidth]{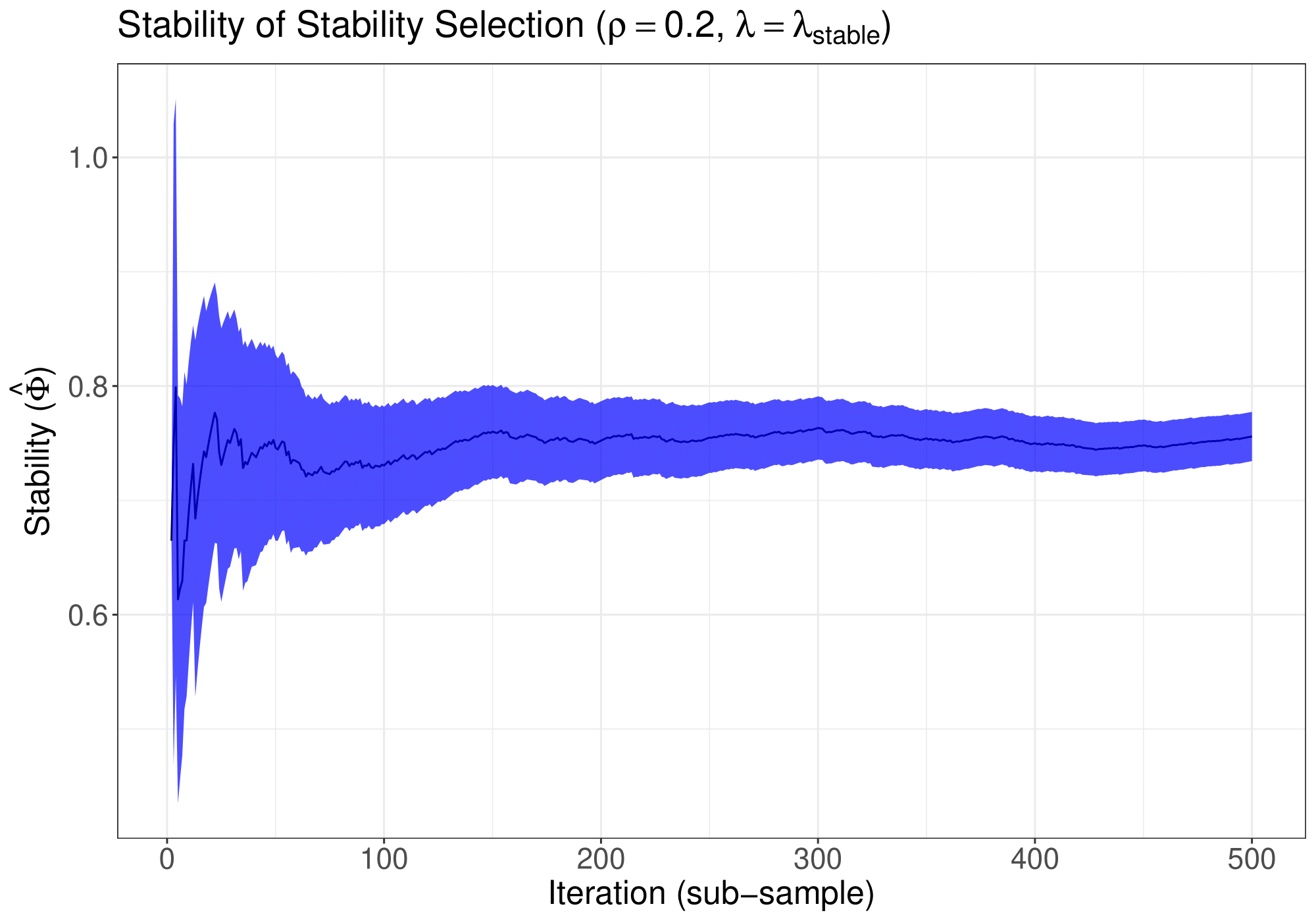}
        \caption{$\rho = 0.2$}
        \label{fig:convergence1}
    \end{subfigure}
    \vfill
    \begin{subfigure}{0.6\textwidth}
        \centering
        \includegraphics[width=\textwidth]{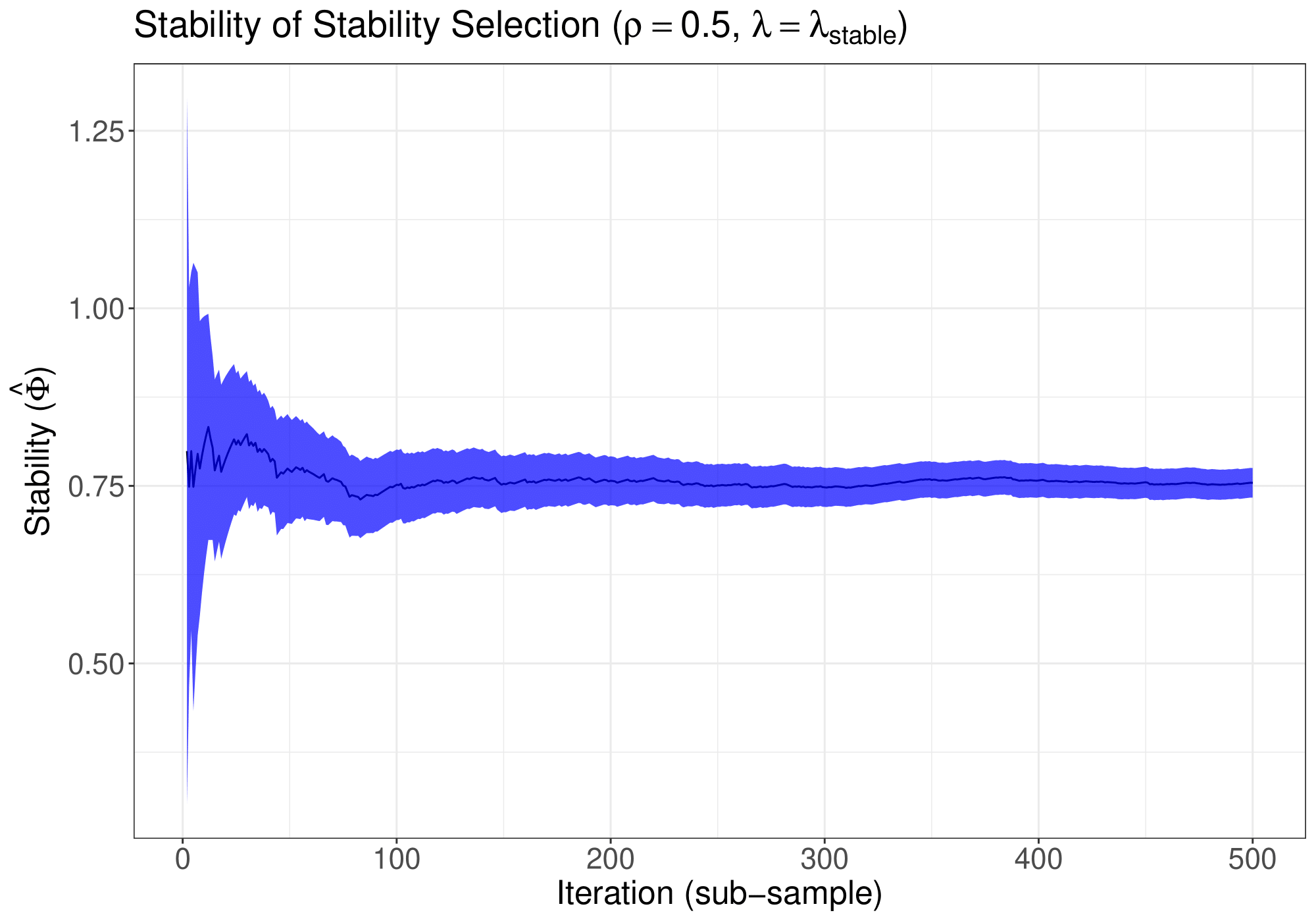}
        \caption{$\rho = 0.5$}
        \label{fig:convergence2}
    \end{subfigure}
    \vfill
    \begin{subfigure}{0.6\textwidth}
        \centering
        \includegraphics[width=\textwidth]{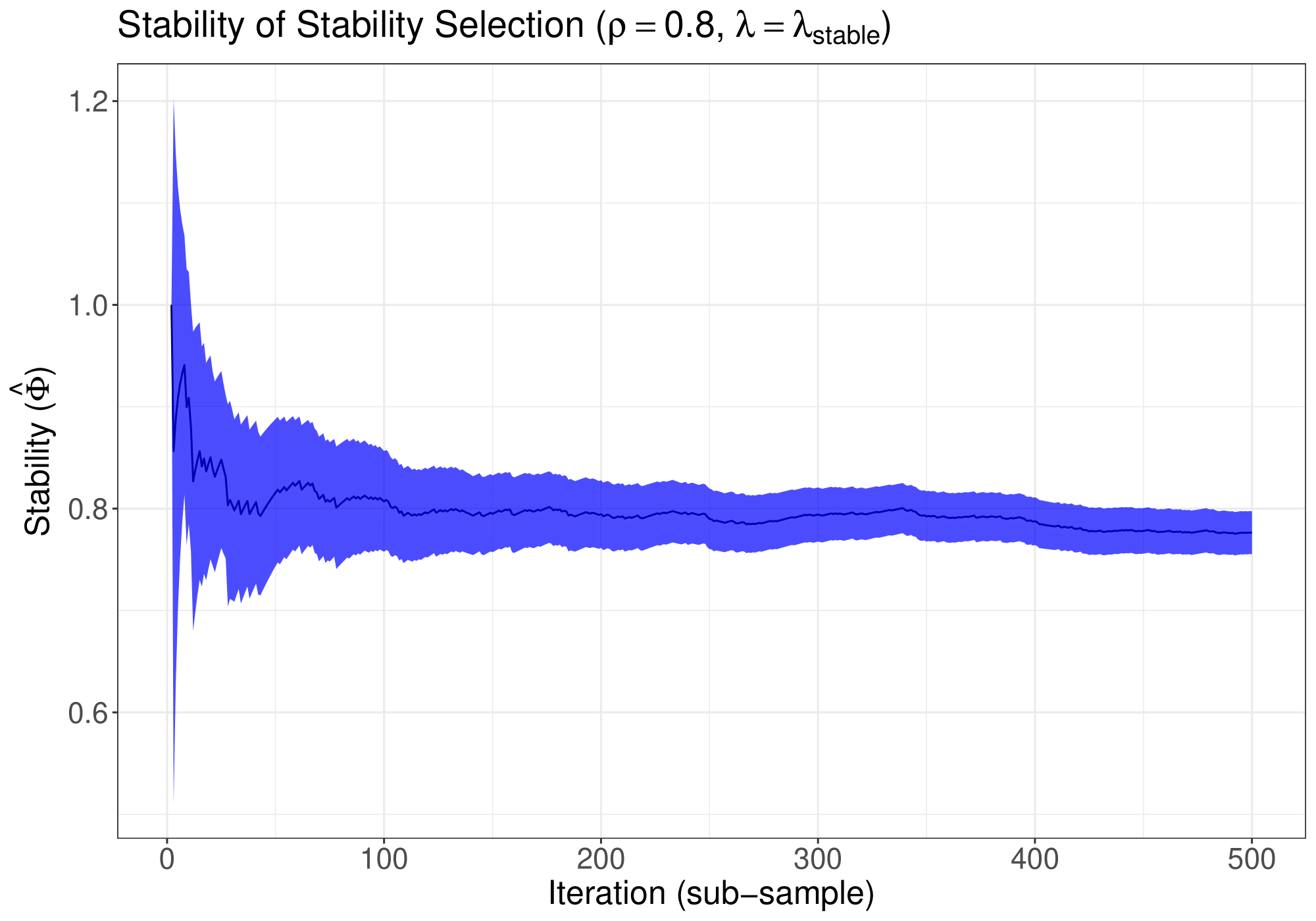} 
        \caption{$\rho = 0.8$}
        \label{fig:convergence3}
    \end{subfigure}
    \caption{Selection stability of stability selection over sequential sub-sampling on the synthetic data}
    \label{fig:three_convergence}
\end{figure}

To demonstrate the convergence of the stability estimator $\hat{\Phi}(\cdot)$, for each correlation scenario, we used $\lambda_{\text{stable}}$ and calculated the stability of the selection matrix $M(\lambda_{\text{stable}})$ through sequential sub-sampling, using Equation~\eqref{eqn: phi}. The stability values across iterations are shown in Figure \ref{fig:three_convergence}. The blue shading around the line represents the $95\%$ confidence interval for $\hat{\Phi}(M(\lambda_{\text{stable}}))$. As can be observed, the interval narrows with increasing iterations, indicating a reduction in the estimator's uncertainty. Figure \ref{fig:three_convergence} reveals that, across all three scenarios, after approximately $B^{*} \approx 200$ iterations, the stability estimator converges, with no significant change thereafter in its value. By monitoring the stability trajectory throughout the process, valuable insight can be gained in determining the optimal cut-off point, that is, the optimal number of sub-samples in terms of stability of the results.

\subsubsection*{Riboflavin Data}

Next, we apply stable stability selection to the riboflavin dataset. As described in Section \ref{s3}, the dataset consists of $p = 4,088$ genes, the main objective being to identify the key genes that are associated with riboflavin production. As before, we choose the number of sub-samples $B = 500$. The predictor variables are standardized using the \texttt{scale} function prior
to being input into Lasso. In this problem, $\lambda_{\text{stable}}$ does not exist, so we use $\lambda_{\text{stable-1sd}}$ instead. After $B$ iterations, four genes \texttt{YXLD\_at} ($\hat{\Pi}_{\texttt{YXLD\_at}}^{\lambda_{\text{stable-1sd}}} = 0.606$), \texttt{YOAB\_at} ($\hat{\Pi}_{\texttt{YOAB\_at}}^{\lambda_{\text{stable-1sd}}} = 0.558$), \texttt{LYSC\_at} ($\hat{\Pi}_{\texttt{LYSC\_at}}^{\lambda_{\text{stable-1sd}}} = 0.540$), and \texttt{YCKE\_at} ($\hat{\Pi}_{\texttt{YCKE\_at}}^{\lambda_{\text{stable-1sd}}} = 0.532$) have selection frequencies greater than $0.5$.

Figure \ref{fig:Riboflavin} illustrates the stability of stability selection results when applied to the riboflavin dataset. As shown in Figure \ref{fig:Riboflavin1}, $\lambda_{\text{stable}}$ does not exist in this example, as the stability values do not surpass the $0.75$ threshold. Figure \ref{fig:Riboflavin2}, generated with $\lambda_{\text{stable-1sd}}$, indicates the convergence slightly above $0.2$ after about $200$ iterations. These results imply that stability selection with Lasso exhibits poor selection stability on this dataset.

\subsubsection*{Affymetrix Rat Genome 230 2.0 Array}

Finally, we apply our methodology to the rat microarray data. As mentioned in Section \ref{s3}, the data consists of $p = 3,083$ gene probes. The main aim for this data is to identify probes that are associated with the TRIM32 probe. Again, we choose the number of sub-samples $B = 500$. The predictor variables are standardized using the \texttt{scale} function prior
to being input into Lasso. Due to the instability of the results, $\lambda_{\text{stable}}$ does not exist in this problem. After $B$ iterations, three probes have selection frequencies greater than $0.5$: probe \texttt{1390539\_at} ($\hat{\Pi}_{\texttt{1390539\_at}}^{\lambda_{\text{stable-1sd}}} = 0.640$), probe \texttt{1389457\_at} ($\hat{\Pi}_{\texttt{1389457\_at}}^{\lambda_{\text{stable-1sd}}} = 0.570$), and probe \texttt{1376747\_at} ($\hat{\Pi}_{\texttt{1376747\_at}}^{\lambda_{\text{stable-1sd}}} = 0.564$). 

As shown in Figure \ref{fig:Rat}, similar to the results of the riboflavin dataset, the results cannot be considered stable. Specifically, Figure \ref{fig:Rat1} reveals that $\lambda_{\text{stable}}$ does not exist, as the stability values do not exceed the threshold of $0.75$. For Figure \ref{fig:Rat2}, we used $\lambda_{\text{stable-1sd}}$. Figure \ref{fig:Rat2} shows that the stability values converged to approximately $0.15$.

The results illustrated in figures \ref{fig:Riboflavin} and \ref{fig:Rat} indicate that while stability selection is a valuable approach, more can be revealed by focusing on the stability of stability selection results. As the complexity of the data (in terms of interdependencies between variables, number of variables subject to selection, etc.) increases, the stability values demonstrate a significant decline. Therefore, it is essential to evaluate the stability of its selections before drawing any conclusions about the relevance of variables based on the results of stability selection.

\subsubsection*{Stability and Accuracy from a Pareto Front Perspective}

\citet{nogueira2018stability} suggested that stability and accuracy can be analyzed using the concept of the Pareto front \citep{pareto1896cours}, which identifies regularization values that are not dominated by any other in terms of both criteria. A regularization value is considered Pareto optimal if no other value on the regularization grid offers both higher stability and higher accuracy. We call this approach stability-accuracy selection, which seeks to balance these two metrics. 

As an experiment, in our synthetic data analysis we consider the negative of the mean squared error (-MSE) as a measure of prediction accuracy. Given that a Pareto optimal solution is not necessarily unique, we select the Pareto solution that maximizes the sum of accuracy and stability. In Figure \ref{fig:pareto1}, $\lambda_{\text{Pareto}}$ is located close to $\lambda_{\text{stable}}$, while in Figures \ref{fig:pareto2} and \ref{fig:pareto3}, they coincide exactly. Interestingly, $\lambda_{\text{stable}}$ is also a Pareto solution in Figure \ref{fig:pareto1}. It is now pertinent to discuss the relationship between $\lambda_{\text{stable}}$ and Pareto optimality.

\begin{corollary}\label{corollary1}
    Let $\lambda_{\text{stable}}$ be defined as in Equation~\eqref{eqn: lambda_stable} and assume that it exists. If the stability curve is non-decreasing up to $\lambda_{\text{stable}}$, and the loss function is non-decreasing after $\lambda_{\text{stable}}$, then $\lambda_{\text{stable}}$ is a Pareto optimal solution.
\end{corollary}

\begin{proof}
    Since the stability curve is non-decreasing prior to $\lambda_{\text{stable}}$, any regularization value $\lambda < \lambda_{\text{stable}} \in \Lambda$ is at most as stable as $\lambda_{\text{stable}}$. Therefore, these values do not dominate $\lambda_{\text{stable}}$ in terms of stability. Similarly, since the loss function is non-decreasing after $\lambda_{\text{stable}}$, any regularization value $\lambda > \lambda_{\text{stable}} \in \Lambda$ is at most as accurate as $\lambda_{\text{stable}}$. Hence, these values also do not dominate $\lambda_{\text{stable}}$ in terms of accuracy. Since $\lambda_{\text{stable}}$ is not dominated by any values less than it in terms of stability nor by any values greater than it in terms of accuracy, we conclude that $\lambda_{\text{stable}}$ constitutes a Pareto optimal solution.
\end{proof}

\begin{figure}[H]
    \centering
    \begin{subfigure}{0.8\textwidth}
        \centering
        \includegraphics[width=\textwidth]{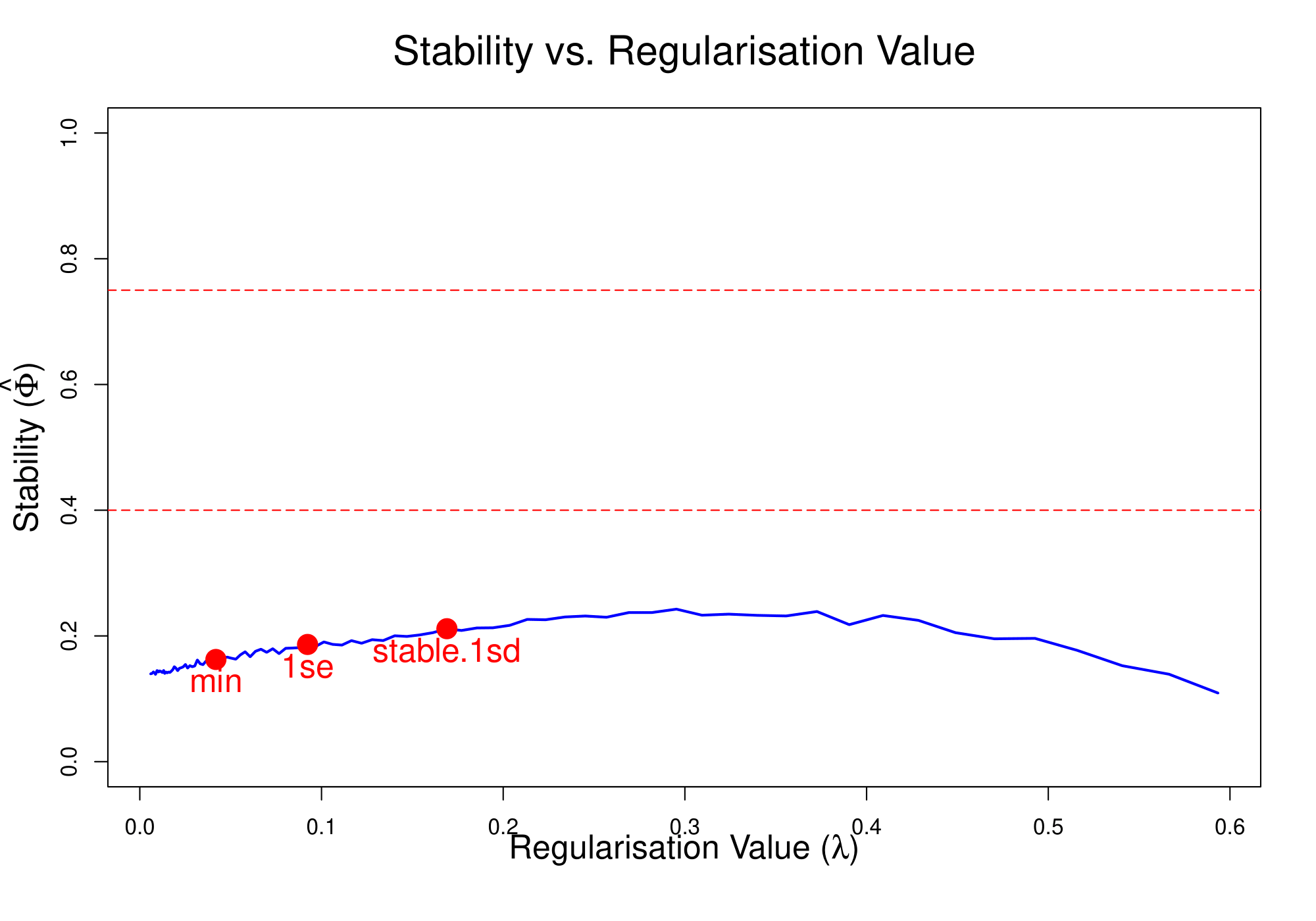}
        \caption{}
        \label{fig:Riboflavin1}
    \end{subfigure}
    \vfill
    \begin{subfigure}{0.8\textwidth}
        \centering
        \includegraphics[width=\textwidth]{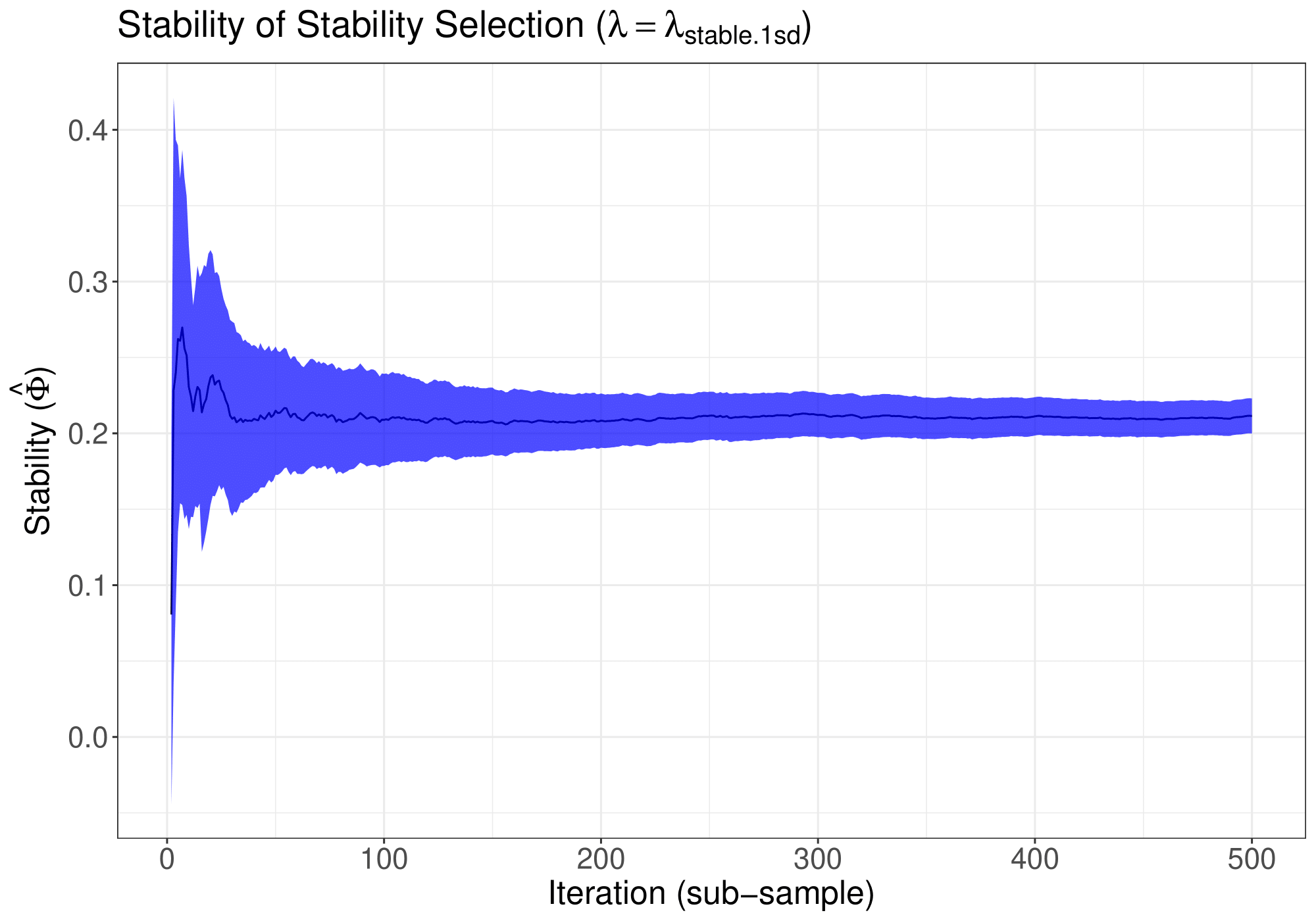}
        \caption{}
        \label{fig:Riboflavin2}
    \end{subfigure}
    \caption{Selection stability of stability selection on riboflavin data}
    \label{fig:Riboflavin}
\end{figure}

\begin{figure}[H]
    \centering
    \begin{subfigure}{0.8\textwidth}
        \centering
        \includegraphics[width=\textwidth]{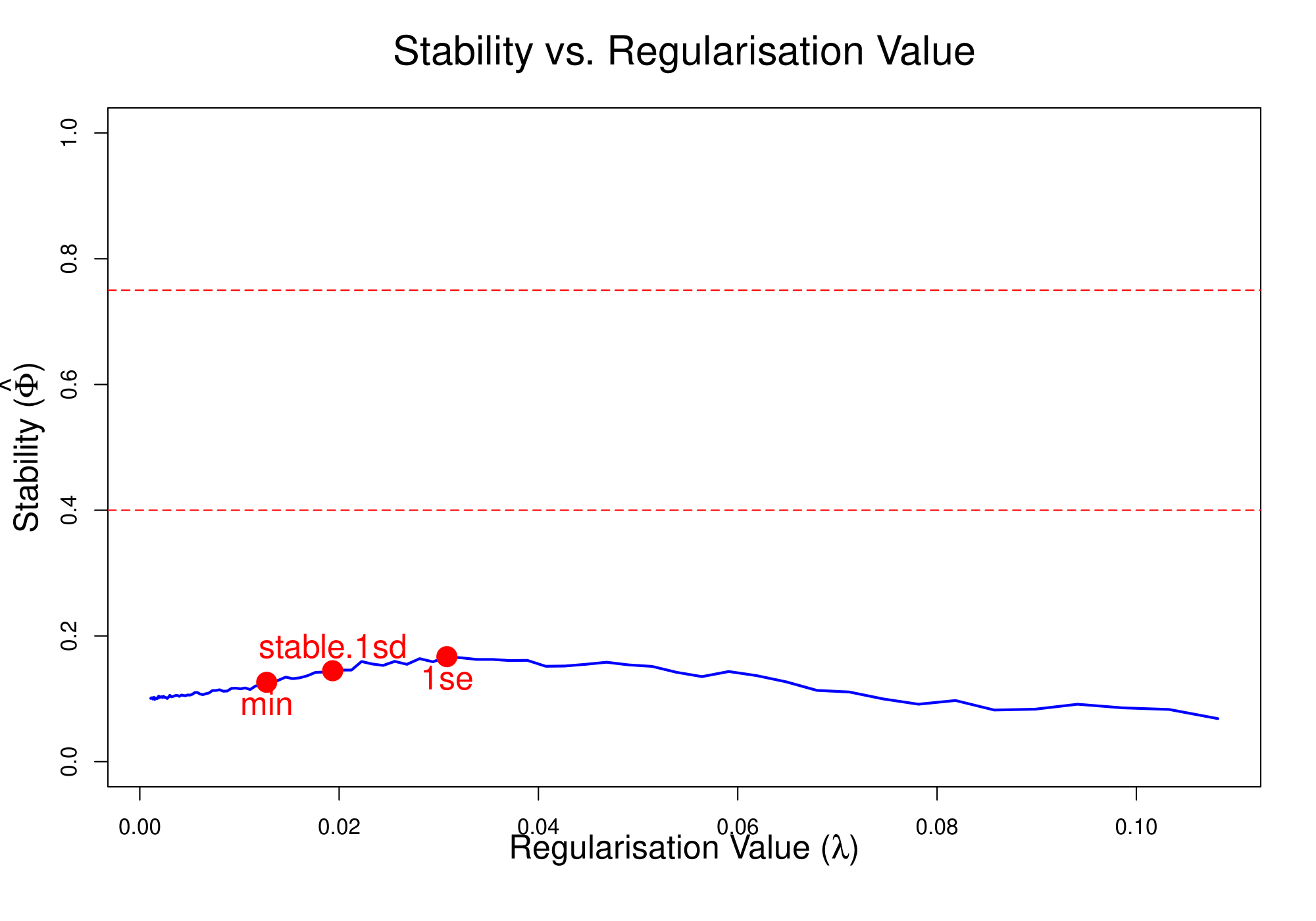}
        \caption{}
        \label{fig:Rat1}
    \end{subfigure}
    \vfill
    \begin{subfigure}{0.8\textwidth}
        \centering
        \includegraphics[width=\textwidth]{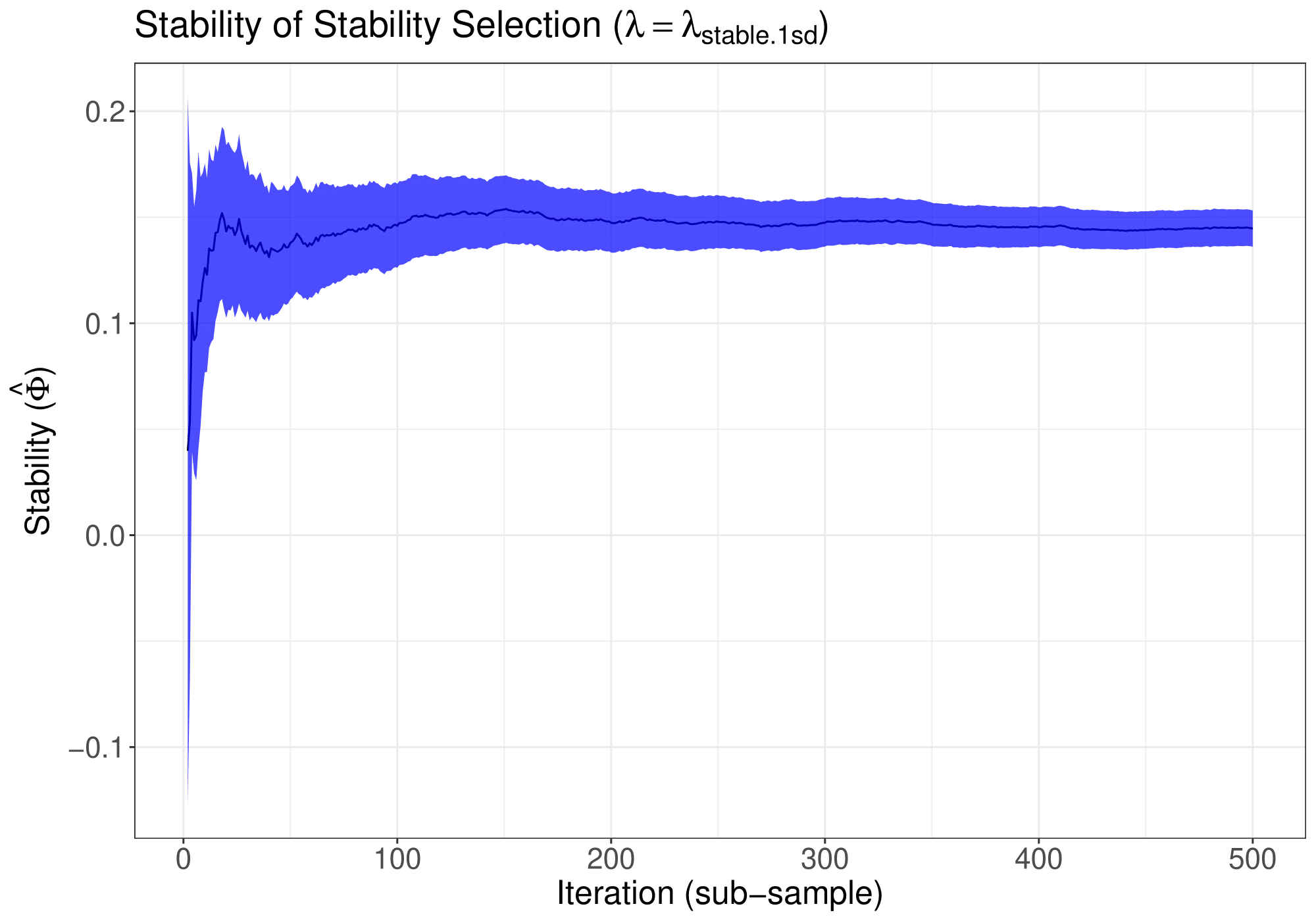}
        \caption{}
        \label{fig:Rat2}
    \end{subfigure}
    \caption{Selection stability of stability selection on rat microarray data}
    \label{fig:Rat}
\end{figure}

\begin{figure}[H]
    \centering
    \begin{subfigure}{0.6\textwidth}
        \centering
        \includegraphics[width=\textwidth]{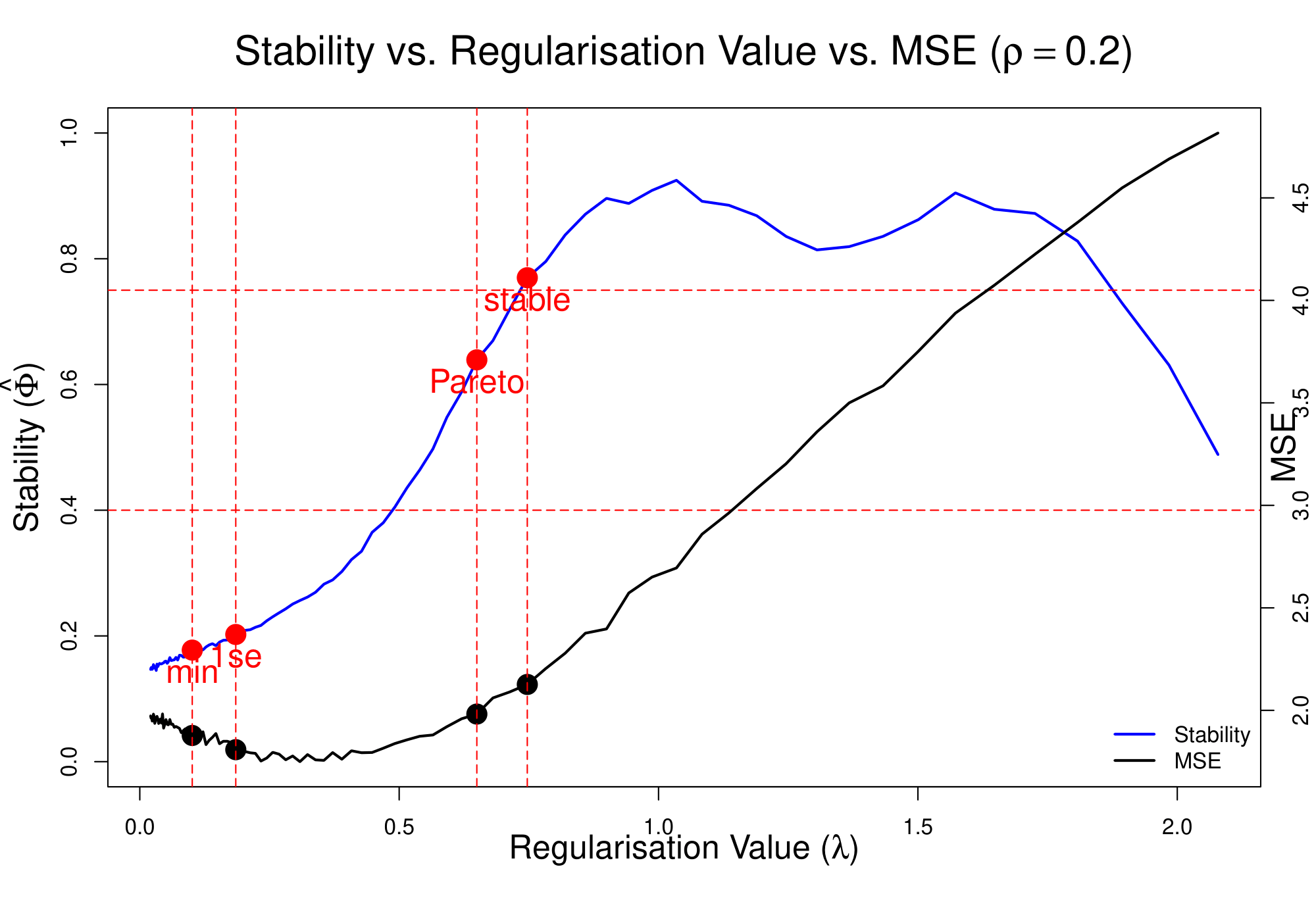}
        \caption{$\rho = 0.2$}
        \label{fig:pareto1}
    \end{subfigure}
    \vfill
    \begin{subfigure}{0.6\textwidth}
        \centering
        \includegraphics[width=\textwidth]{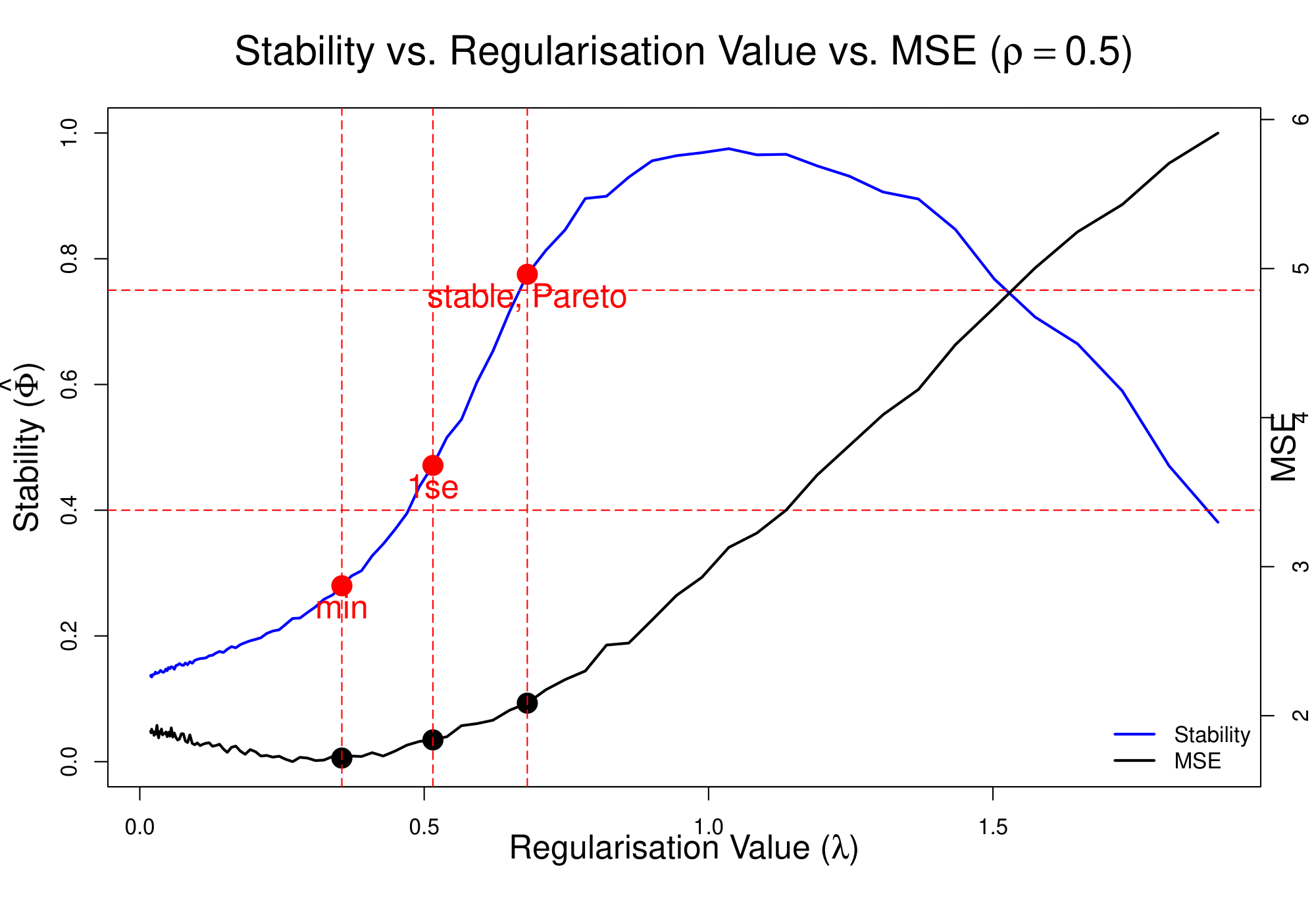}
        \caption{$\rho = 0.5$}
        \label{fig:pareto2}
    \end{subfigure}
    \vfill
    \begin{subfigure}{0.6\textwidth}
        \centering
        \includegraphics[width=\textwidth]{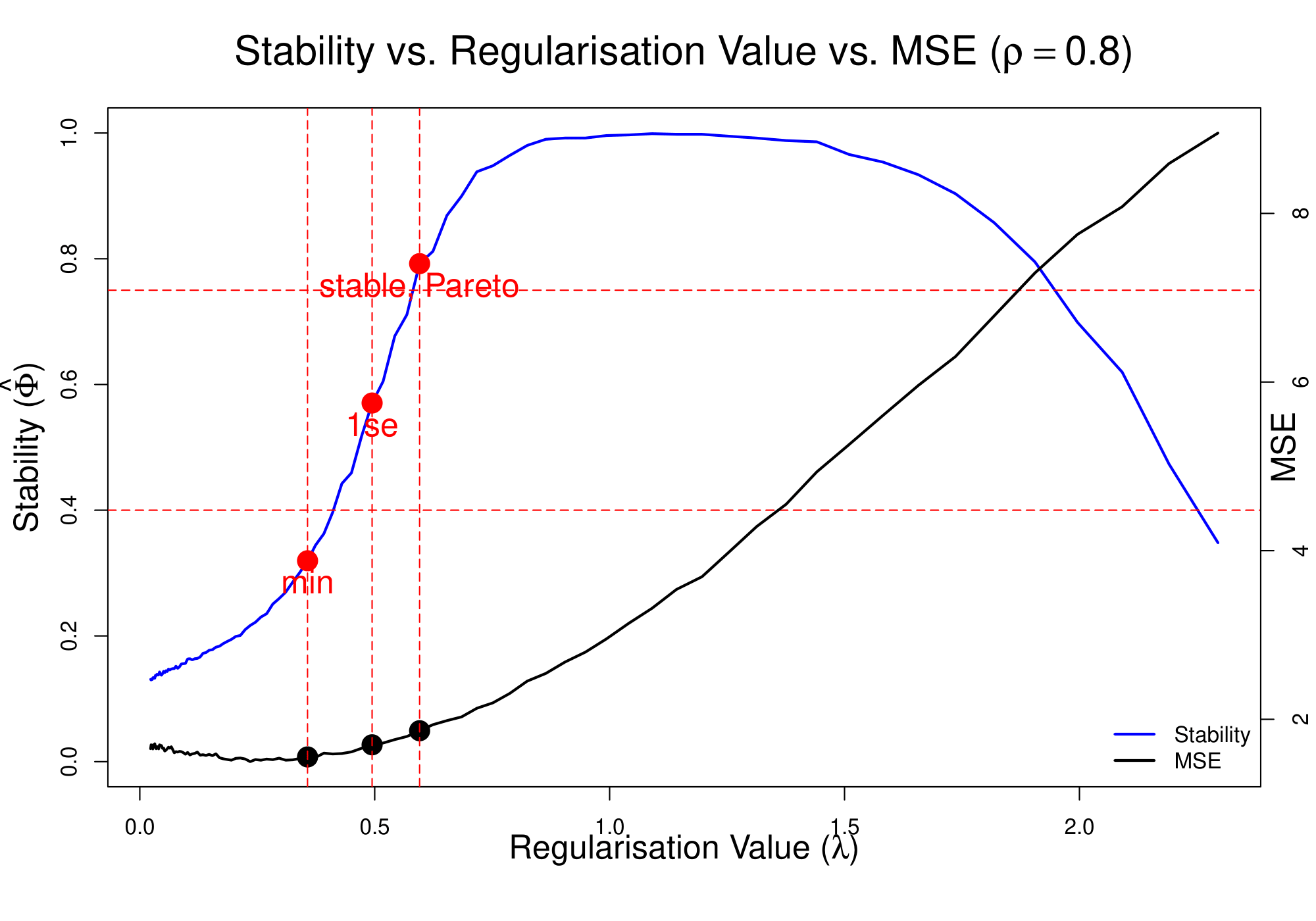} 
        \caption{$\rho = 0.8$}
        \label{fig:pareto3}
    \end{subfigure}
    \caption{$\lambda_{\text{min}}$, $\lambda_{\text{1se}}$, $\lambda_{\text{stable}}$, and $\lambda_{\text{Pareto}}$ over the regularization grid, the latter two are identical when $\rho$ is 0.5 or 0.8}
    \label{fig:three_pareto}
\end{figure}

The two assumptions in Corollary \ref{corollary1} are both justifiable. Increasing the regularization value from zero is expected to allow the model to achieve an optimal balance of sparsity by shrinking irrelevant variables, thereby enhancing stability. The regularization value $\lambda_{\text{stable}}$ is intended to signify the point where the model attains high stability with minimal loss of accuracy. Consequently, we anticipate that the stability curve exhibits a non-decreasing trend prior to reaching $\lambda_{\text{stable}}$. 

Regarding the second assumption, it is expected that increasing regularization helps the model eliminate irrelevant variables, improving its ability to predict the response variable. The regularization values $\lambda_{\text{min}}$ and $\lambda_{\text{1se}}$ are designed to capture this behavior, marking the points where regularization is most effective in terms of predictive ability. Beyond these points, we anticipate that the loss curve will be non-decreasing as a consequence of model over-shrinkage. Since $\lambda_{\text{stable}}$ is intended to trade a small amount of accuracy for greater stability, we expect $\lambda_{\text{stable}}$ to occur after $\lambda_{\text{1se}}$, that is, $\lambda_{\text{stable}} > \lambda_{\text{1se}} > \lambda_{\text{min}}$, which implies that the loss curve is expected to be non-decreasing after $\lambda_{\text{stable}}$. 

Thus, according to Corollary \ref{corollary1}, $\lambda_{\text{stable}}$ represents a stability-accuracy solution for the problem of regularization tuning, ensuring high stability while minimizing loss in prediction accuracy.

Although the methodology is presented within the context of the Lasso, its extension to other regularization-based variable selection techniques remains a natural and conceptually supported direction, as suggested by \citet{10.5555/2567709.2567772}.

\section*{Competing Interests}
The authors declare that they have no conflict of interest.

\section*{Author Contributions Statement}
 
Mahdi Nouraie was responsible for drafting the manuscript, the development of the research methodology and for writing the computer code used throughout. Samuel Muller provided critical feedback on the content of the manuscript, refining the clarity and scope of the manuscript and the computer code.  

\section*{Data Availability}
The riboflavin dataset is accessible via the \texttt{hdi} package in R \citep{hdi-package}. The rat microarray data can be obtained from the National Center for Biotechnology Information (NCBI) website at \url{www.ncbi.nlm.nih.gov}, under accession number GSE5680.

The source code used for the paper is accessible through the following GitHub repository: \url{https://github.com/MahdiNouraie/Stable-Stability-Selection}. Furthermore, the \texttt{stabplot} R package, which facilitates the use of the methodology introduced in this paper, is available through \url{https://github.com/MahdiNouraie/stabplot}.

\section*{Acknowledgments}
Mahdi Nouraie was supported by the Macquarie University Research Excellence Scholarship (20213605). Samuel Muller was supported by the Australian Research Council Discovery Project Grant (DP230101908). We gratefully acknowledge Dr. Connor Smith for his guidance and support as co-supervisor of Mahdi Nouraie’s PhD research.

\bibliographystyle{plainnat}
\bibliography{citation} 
\end{document}